\documentclass[11pt]{myclass}
\usepackage{euscript}
\usepackage{wrapfig}
\usepackage{amssymb}
\usepackage{epsfig}
\usepackage{xspace}
\usepackage{color}

\usepackage{hyperref}
\usepackage{amsmath}
\usepackage{amsthm}
\usepackage{color}
\usepackage{url}

\newcommand{\eps}{\varepsilon}

\newcommand{\Eu}[1]{\ensuremath{\EuScript{#1}}}

\newcommand{\bl}[1]{\ensuremath{\mathbb{#1}}}

\newcommand{\E}{\textbf{\textsf{E}}}

\renewcommand{\Pr}{\textbf{\textsf{Pr}}}
\newcommand{\disc}{\textsf{disc}}
\newcommand{\cert}{\textsf{cert}}

\newcommand{\RR}{\textsf{RR}\xspace}
\newcommand{\RC}{\textsf{RC}\xspace}
\newcommand{\RE}{\textsf{RE}\xspace}
\newcommand{\RQ}{\textsf{RQ}\xspace}

\newcommand{\REd}{\RE\text{-}\disc}

\newcommand{\denselist}{\vspace{-.1in} \itemsep -2pt\parsep=-1pt\partopsep -2pt}

\newcommand{\Paragraph}[1]{\paragraph{\sffamily\textbf{#1.}}}

\title{Range Counting Coresets for Uncertain Data} 

\author{Amirali Abdullah \\ {\small University of Utah} \\ {\small \texttt{amirali@cs.utah.edu}}
\and 
Samira Daruki \\ {\small University of Utah} \\ {\small \texttt{daruki@cs.utah.edu}}
\and 
Jeff M. Phillips \\ {\small University of Utah} \\ {\small \texttt{jeffp@cs.utah.edu}}}

\begin{document}
\maketitle

\begin{abstract}
We study coresets for various types of range counting queries on uncertain data.  In our model each uncertain point has a probability density describing its location, sometimes defined as $k$ distinct locations.  Our goal is to construct a subset of the uncertain points, including their locational uncertainty, so that range counting queries can be answered by just examining this subset.  
We study three distinct types of queries.
\RE queries return the expected number of points in a query range.  
\RC queries return the number of points in the range with probability at least a threshold.
\RQ queries returns the probability that fewer than some threshold fraction of the points are in the range.
In both \RC and \RQ coresets the threshold is provided as part of the query.  
And for each type of query we provide coreset constructions with approximation-size tradeoffs.  
We show that random sampling can be used to construct each type of coreset, and we also provide significantly improved bounds using discrepancy-based approaches on axis-aligned range queries.  
\end{abstract}

\keywords{uncertain data, coresets, discrepancy}

\section{Introduction}
A powerful notion in computational geometry is the \emph{coreset}~\cite{AHV04,AHV07,BC03,VC71}.  Given a large data set $P$ and a family of queries $\Eu{A}$, then an \emph{$\eta$-coreset} is a subset $S \subset P$ such that for all $r \in \Eu{A}$ that $\|r(P) - r(S)\| \leq \eta$ (note the notion of distance $\|  \cdot \|$ between query results is problem specific and is intentionally left ambiguous for now).  Initially used for smallest enclosing ball queries~\cite{BC03} and perhaps most famous in geometry for extent queries as $\eta$-kernels~\cite{AHV04,AHV07}, the coreset is now employed in many other problems such as clustering~\cite{BHP02} and density estimation~\cite{VC71}.  
Techniques for constructing coresets are becoming more relevant in the era of big data; they summarize a large data set $P$ with a proxy set $S$ of potentially much smaller size that can guarantee error for certain classes of queries.  They also shed light onto the limits of how much information can possibly be represented in a small set of data.  

In this paper we focus on a specific type of coreset called an $\eta$-sample~\cite{VC71,peled,CM96} that can be thought of as preserving density queries and that has deep ties to the basis of learning theory~\cite{AB99}.  Given a set of objects $X$ (often $X \subset \bl{R}^d$ is a point set) and a family of subsets $\Eu{A}$ of $X$, then the pair $(X,\Eu{A})$ is called an \emph{range space}.  Often $\Eu{A}$ are specified by containment in geometric shapes, for instance as all subsets of $X$ defined by inclusion in any ball, any half space, or any axis-aligned rectangle.  Now an $\eta$-sample of $(X,\Eu{A})$ is a single subset $S \subset X$ such that 
\[
\max_{r \in \Eu{A}} \left| \frac{|X \cap r|}{|X|} - \frac{|S \cap r|}{|S|} \right| \leq \eta.
\]
For any query range $r \in \Eu{A}$, subset $S$ approximates the relative density of $X$ in $r$ with error at most $\eta$.  

\Paragraph{Uncertain points}
Another emerging notion in data analysis is modeling uncertainty in points.  There are several formulations of these problems where each point $p \in P$ has an \emph{independent} probability distribution $\mu_p$ describing its location and such a point is said to have \emph{locational uncertainty}.  
\emph{Imprecise points} (also called \emph{deterministic uncertainty}) model where a data point $p \in P$ could be anywhere within a fixed continuous range and were originally used for analyzing precision errors. The worst case properties of a point set $P$ under the imprecise model have been
well-studied~\cite{gss-egbra-89,gss-cscah-93,bs-ads-04,hm-ticpps-08,ls-dtip-08,nt-teb-00,obj-ue-05,kl-lbbsd-10,k-bmips-08}.
\emph{Indecisive points}  (or \emph{attribute uncertainty} in database literature~\cite{1644250}) model each $p_i \in P$ as being able to take one of $k$ distinct locations $\{p_{i,1}, p_{i,2}, \ldots, p_{i,k}\}$ with possibly different probabilities, modeling when multiple readings of the same object have been made~\cite{JLP11,MDFW00,CLY09,CG09,ACTY09,ABSHNSW06}.  

We also note another common model of \emph{existential uncertainty} (similar to \emph{tuple uncertainty} in database literature~\cite{1644250} but a bit less general) where the location or value of each $p \in P$ is fixed, but the point may not exist with some probability, modeling false readings~\cite{KCS11a,KCS11b,1644250,CLY09}.

We will focus mainly on the indecisive model of locational uncertainty since it comes up frequently in real-world applications~\cite{MDFW00,ABSHNSW06} (when multiple readings of the same object are made, and typically $k$ is small) and can be used to approximately represent more general continuous representations~\cite{JLP12,Phi08}.  

\subsection{Problem Statement}
Combining these two notions leads to the question: can we create a coreset (specifically for $\eta$-samples) of uncertain input data?  A few more definitions are required to rigorously state this question.  In fact, we develop three distinct notions of how to define the coreset error in uncertain points.  One corresponds to range counting queries, another to querying the mean, and the third to querying the median (actually it approximates the rank for all quantiles).  

For an uncertain point set $P = \{p_1, p_2, \ldots, p_n\}$ with each $p_i = \{p_{i,1}, p_{i,2}, \ldots p_{i,k}\} \subset \bl{R}^d$ we say that $Q \Subset P$ is a \emph{transversal} if $Q \in p_1 \times p_2 \times \ldots \times p_n$.  I.e., $Q = (q_1, q_2, \ldots, q_n)$ is an instantiation of the uncertain data $P$ and can be treated as a ``certain'' point set, where each $q_i$ corresponds to the location of $p_i$.   
$\Pr_{Q \Subset P} [\zeta(Q)]$, (resp. $\E_{Q \Subset P}[\zeta(Q)]$) represents the probability (resp. expected value) of an event $\zeta(Q)$ where $Q$ is instantiated from $P$ according to the probability distribution on the uncertainty in $P$.  

As stated, our goal is to construct a subset of uncertain points $T \subset P$ (including the distribution of each point $p$'s location, $\mu_p$) that preserves specific properties over a family of subsets $(P,\Eu{A})$.  For completeness, the first variation we list cannot be accomplished purely with a coreset as it requires $\Omega(n)$ space.  

\begin{itemize}  \denselist
\item \emph{Range Reporting (\RR) Queries} support queries of a range $r \in \Eu{A}$ and a threshold $\tau$, and return all $p_i \in P$ such that $\Pr_{Q \Subset P} [q_i \in r] \geq \tau$. Note that the fate of each $p_i \in P$ depends on no other $p_j \in P$ where $i \neq j$, so they can be considered independently. Building indexes for this model have been studied~\cite{threshquery,efficientquery,TCXNKP05,ZLTZW12} and effectively solved in $\bl{R}^1$~\cite{ACTY09}.  
\item \emph{Range Expectation (\RE) Queries} consider a range $r \in \Eu{A}$ and report the expected number of uncertain points in $r$, $\E_{Q \Subset P}[ | r \cap Q | ]$. The linearity of expectation allows summing the individual expectations each point $p \in P$ is in $r$. Single queries in this model have also been studied~\cite{JKV07,BDJRV05,JMMV07}.
\item \emph{Range Counting (\RC) Queries} support queries of a range $r \in \Eu{A}$ and a threshold $\tau$, but only return the number of $p_i \in P$ which satisfy $\Pr_{Q \Subset P} [q_i \in r] \geq \tau$.  
The effect of each $p_i \in P$ on the query is separate from that of any other $p_j \in P$ where $i \neq j$.  A random sampling heuristic~\cite{aggregate} has been suggested without proof of accuracy.  
\item \emph{Range Quantile (\RQ) Queries} take a query range $r \in \Eu{A}$, and report the full cumulative density function on the number of points in the range $\Pr_{Q \Subset P} [ |r \cap Q| ]$.  Thus for a query range $r$, this returned structure can produce for any value $\tau \in [0,1]$ the probability that $\tau n$ or fewer points are in $r$.  
Since this is no longer an expectation, the linearity of expectation cannot be used to decompose this query along individual uncertain points.  
\end{itemize}
\vspace{-2mm}
Across all queries we consider, there are two main ways we can approximate the answers.  The first and most standard way is to allow an $\eps$-error (for $0 \leq \eps \leq 1$) in the returned answer for \RQ, \RE, and \RC. 
The second way is to allow an $\alpha$-error in the \emph{threshold} associated with the query itself.  As will be shown, this is not necessary for \RR, \RE, or \RC, but is required to get useful bounds for \RQ.  
Finally, we will also consider probabilistic error $\delta$, demarcating the probability of failure in a randomized algorithm (such as random sampling).  
We strive to achieve these approximation factors with a small size coreset $T \subset P$ as follows: 
\begin{itemize}\denselist
\item [\RE:]
For a given range $r$, let $r(Q) = |Q \cap r|/|Q|$, and let 
$\begin{displaystyle}
E_{r(P)} = \E_{Q \Subset P}[r(Q)].
\end{displaystyle}$
$T \subset P$ is an \emph{$\eps$-\RE coreset} of $(P, \Eu{A})$ if for all queries $r \in \Eu{A}$  we have
$\begin{displaystyle}
\left| E_{r(P)} - E_{r(T)} \right| \leq \eps.  
\end{displaystyle}$
\item [\RC:]
For a range $r \in \Eu{A}$, let 
$
G_{P,r}(\tau) = \frac{1}{|P|} \big|\big\{p_i \in P \mid \Pr_{Q \Subset P}[q_i \in r] \geq \tau\big\} \big|
$
be the fraction of points in $P$ that are in $r$ with probability at least some threshold $\tau$.  
Then $T \subset P$ is an \emph{$\eps$-\RC coreset} of $(P,\Eu{A})$ if for all queries $r \in \Eu{A}$ and all $\tau \in [0,1]$ we have
$\begin{displaystyle}
\left| G_{P,r}(\tau) - G_{T,r}(\tau) \right| \leq \eps.  
\end{displaystyle}$
\item [\RQ:]
For a range $r \in \Eu{A}$, let 
$
F_{P,r}(\tau) = \Pr_{Q \Subset P} [r(Q) \leq \tau] = \Pr_{Q \Subset P} \hspace{-1mm}\left[\frac{|Q \cap r|}{|Q|} \leq \tau\right]
$
be the probability that at most a $\tau$ fraction of $P$ is in $r$.  
Now $T \subset P$ is an \emph{($\eps,\alpha$)-\RQ coreset} of $(P,\Eu{A})$ if for all $r \in \Eu{A}$  and $\tau \in [0,1]$ there exists a $\gamma \in [\tau - \alpha, \tau + \alpha]$ such that 
$\begin{displaystyle}
\left| F_{P,r}(\tau) - F_{T,r}(\gamma) \right| \leq \eps.
\end{displaystyle}$
In such a situation, we also say that $F_{T,r}$ is an \emph{$(\eps,\alpha)$-quantization} of $F_{P,r}$.  
\end{itemize}

\vspace{-0mm}
A natural question is whether we can construct a ($\eps,0$)-\RQ coreset where there is not a secondary $\alpha$-error term on $\tau$.  We demonstrate that there are no useful non-trivial bounds on the size of such a coreset.    

When the $(\eps,\alpha)$-quantization $F_{T,r}$ need not be explicitly represented by a coreset $T$, then L\"offler and Phillips~\cite{LP09,JLP11} show a different small space representation that can replace it in the above definition of an $(\eps,\alpha)$-\RQ coreset with probability at least $1-\delta$.  First randomly create $m = O((1/\eps^2) \log (1/\delta))$ transversals $Q_1, Q_2, \ldots, Q_m$, and for each transversal $Q_i$ create an $\alpha$-sample $S_i$ of $(Q_i, \Eu{A})$.  Then to satisfy the requirements of $F_{T,r}(\tau)$, there exists some $\gamma \in [\tau-\alpha, \tau + \alpha]$ such that we can return $(1/m) |\{S_i \mid r(S_i) \leq \gamma\}|$, and it will be within $\eps$ of $F_{P,r}(\tau)$.  
However, this is subverting the attempt to construct and understand a coreset to answer these questions.  A coreset $T$ (our goal) can be used as proxy for $P$ as opposed to querying $m$ distinct point sets.  This alternate approach also does not shed light into how much information can be captured by a small size point set, which is provided by bounds on the size of a coreset.  

\Paragraph{Simple example}
We illustrate a simple example with $k=2$ and $d=1$, where $n=10$ and the $nk =20$ possible locations of the $10$ uncertain points are laid out in order:
\begin{align*}
p_{1,1} &< p_{2,1} < p_{3,1} < p_{4,1} < p_{5,1} <  p_{6,1} < p_{3,2} < p_{7,1} 
<  
p_{8,1} < p_{8,2} < p_{9,1} < p_{10,1} < p_{5,2} < p_{10,2} <  p_{2,2}
\\&< 
 p_{9,2}  < p_{7,2} <  p_{4,2} < p_{6,2} <  p_{1,2}.
\end{align*}

We consider a coreset $T \subset P$ that consists of the uncertain points $T = \{p_1, p_3, p_5, p_7, p_9\}$. 
Now consider a specific range $r \in \Eu{I}_+$, a one-sided interval that contains $p_{5,2}$ and smaller points, but not $p_{10,2}$ and larger points.  
We can now see that $F_{T,r}$ is an $(\eps'=0.1016, \alpha=0.1)$-quantization of $F_{P,r}$ in Figure \ref{fig:cdf-10}; this follows since at $F_{P,r}(0.75) = 0.7734$ either $F_{T,r}(x)$ is at most $0.5$ for $x \in [0.65,0.8)$ and is at least $0.875$ for $x \in [0.8,0.85]$.  
Also observe that 
\[
\left|E_{r(P)} - E_{r(T)}\right| = \left|\frac{13}{20} - \frac{7}{10} \right| = \frac{1}{20} = \eps.
\]  
When these errors (the $(\eps',\alpha)$-quantization and $\eps$-error) hold for \emph{all} ranges in some range space, then $T$ is an $(\eps',\alpha)$-\RQ coreset or $\eps$-\RC coreset, respectively.  

To understand the error associated with an \RC coreset, also consider the threshold $\tau = 2/3$ with respect to the range $r$.  Then in range $r$, $2/10$ of the uncertain points from $P$ are in $r$ with probability at least $\tau = 2/3$ (points $p_3$ and $p_8$).  Also $1/5$ of the uncertain points from $T$ are in $r$ with probability at least $\tau = 2/3$ (only point $p_3$).  So there is $0$ \RC error for this range and threshold.

\begin{figure}[t]
\centering
\includegraphics[width=0.6\linewidth]{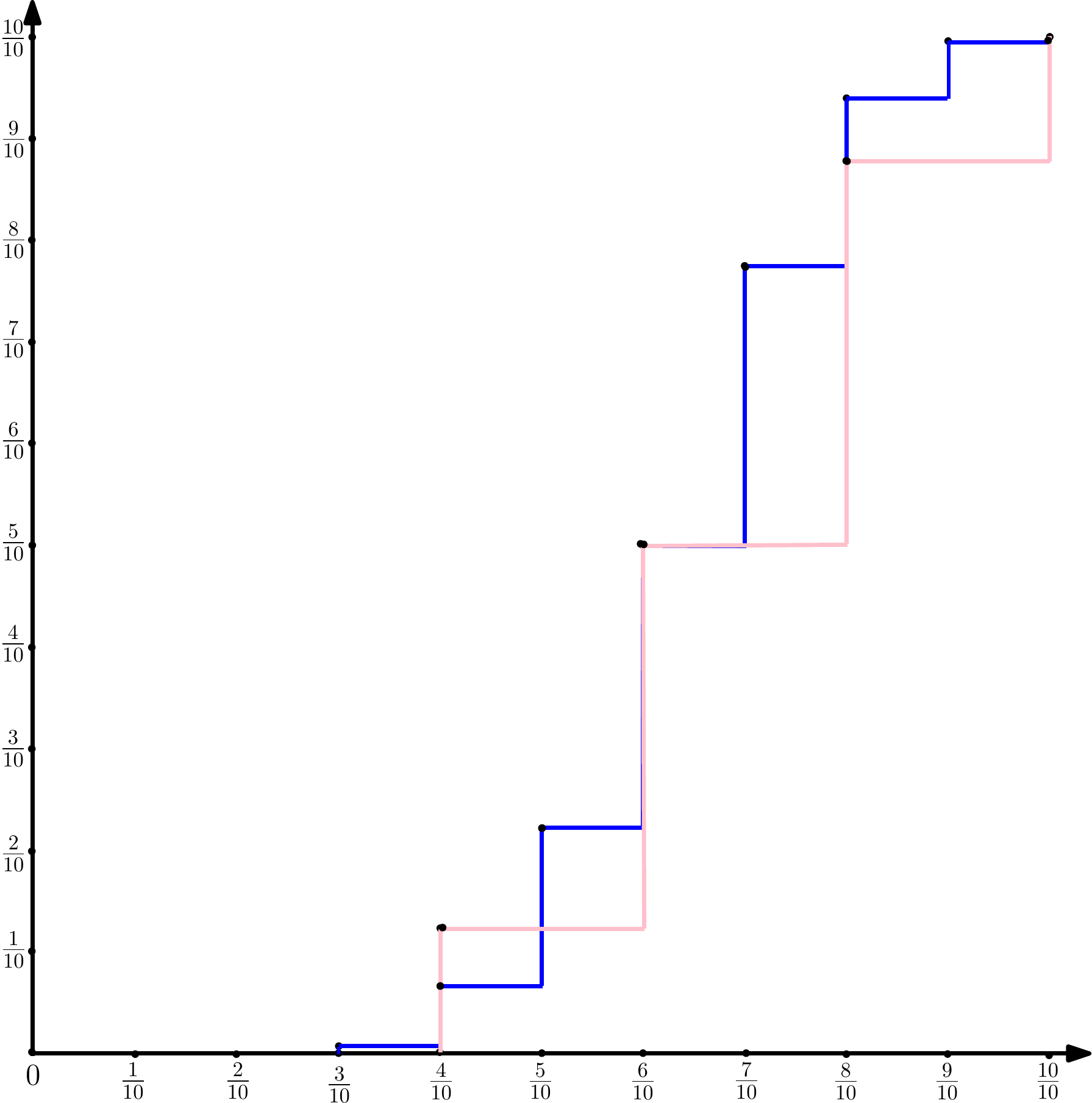}
\caption{Example cumulative density functions  ($F_{T,r}$, in red with fewer steps, and $F_{P,r}$, in blue with more steps) on uncertain point set $P$ and a coreset $T$ for a specific range.}
\label{fig:cdf-10}
\end{figure}

\subsection{Our Results}

We provide the first results for \RE-, \RC-, and \RQ-coresets with guarantees.  In particular we show that a random sample $T$ of size $O((1/\eps^2) (\nu + \log(1/\delta))$ with probability $1-\delta$ is an $\eps$-\RC coreset for any family of ranges $\Eu{A}$ whose associated range space has VC-dimension $\nu$.  
Otherwise we enforce that each uncertain point has $k$ possible locations, then a sample $T$ of size $O((1/\eps^2) (\nu + \log(k/\delta))$ suffices for an $\eps$-\RE coreset.  

Then we leverage discrepancy-based techniques~\cite{Mat99,Cha01} for some specific families of ranges $\Eu{A}$, to improve these bounds to $O((1/\eps) \textsf{poly}(k, \log(1/\eps)))$.  This is an important improvement since $1/\eps$ can be quite large (say $100$ or more), while $k$, interpreted as the number of readings of a data point, is small for many applications (say $5$).  
In $\bl{R}^1$, for one-sided ranges we construct $\eps$-\RE and $\eps$-\RC coresets of size $O((\sqrt{k}/\eps) \log(k/\eps))$.
For axis-aligned rectangles in $\bl{R}^d$ we construct $\eps$-\RE coresets of size $O((\sqrt{k}/\eps) \cdot \log^{\frac{3d-1}{2}} (k/\eps))$ and $\eps$-\RC coresets of size $O((k^{3d+ \frac{1}{2}}/\eps) \log^{6d-\frac{1}{2}} (k/\eps)))$.  
Finally, we show that any $\eps$-\RE coreset of size $t$ is also an $(\eps,\alpha_{\eps,t})$-\RQ coreset with value $\alpha_{\eps,t} = \eps + \sqrt{(1/2t)\ln(2/\eps)}$.  

These results leverage new connections between uncertain points and both discrepancy of permutations and colored range searching that may be of independent interest.  

\section{Discrepancy and Permutations}

The key tools we will use to construct small coresets for uncertain data is discrepancy of range spaces, and specifically those defined on permutations.  
Consider a set $X$, a range space $(X,\Eu{A})$, and a coloring $\chi : X \to \{-1, +1\}$.  
Then for some range $A \in \Eu{A}$, the discrepancy is defined 
$\disc_\chi(X,A) = |\sum_{x \in X \cap A} \chi(x)|$.  We can then extend this to be over all ranges $\disc_\chi(X,\Eu{A}) = \max_{A \in \Eu{A}} \disc_\chi(X,A)$ and over all colorings $\disc(X,\Eu{A}) = \min_\chi \disc_\chi(X,\Eu{A})$.  

Consider a ground set $(P,\Sigma_k)$ where $P$ is a set of $n$ objects, and $\Sigma_k = \{\sigma_1, \sigma_2, \ldots, \sigma_k\}$ is a set of $k$ permutations over $P$ so each $\sigma_j : P \to [n]$.  We can also consider a family of ranges $\Eu{I}_k$ as a set of intervals defined on \emph{one} of the $k$ permutations so $I_{x,y,j} \in \Eu{I}_k$ is defined so $P \cap I_{x,y,j} = \{p \in P \mid x < \sigma_j(p) \leq y\}$ for $x < y \in [0,n]$ and $j \in [k]$. The pair $((P,\Sigma_k),\Eu{I}_k)$ is then a range space, defining a set of subsets of $P$.  

A canonical way to obtain $k$ permutations from an uncertain point set $P = \{p_1, p_2, \ldots, p_n\}$ is as follows.  
Define the \emph{$j$th canonical traversal} of $P$ as the set $P_j = \cup_{i=1}^n p_{i,j}$. When each $p_{i,j} \in \bl{R}^1$, the sorted order of each canonical traversal $P_j$ defines a permutation on $P$ as $\sigma_j(p_i) = |\{p_{i',j} \in P_j \mid p_{i',j} \leq p_{i,j} \} |$, that is $\sigma_j(p_i)$ describes how many locations (including $p_{i,j}$) in the traversal $P_j$ have value less than or equal to $p_{i,j}$.  In other words, $\sigma_j$ describes the sorted order of the $j$th point among all uncertain points.  Then, given an uncertain point set, let the canonical traversals define the \emph{canonical $k$-permutation} as $(P,\Sigma_k)$.  

A geometric view of the permutation range space embeds $P$ as $n$ fixed points in $\bl{R}^k$ and considers ranges which are defined by inclusion in $(k-1)$-dimensional slabs, defined by two parallel half spaces with normals aligned along one of the coordinate axes.  Specifically, the $j$th coordinate of the $i$th point is $\sigma_j(p_i)$, and if the range is on the $j$th permutation, then the slab is orthogonal to the $j$th coordinate axis.  

Another useful construction from an uncertain point set $P$ is the set $P_\cert$ of all locations any point in $P$ might occur.  
Specifically, for every uncertain point set $P$ we can define the corresponding certain point set $P_{\cert} = \bigcup_{i \in [n]} p_i = \bigcup_{j \in [k]} P_j = \bigcup_{i \in [n], j \in [k]} p_{i,j}$. 
We can also extend any coloring $\chi$ on $P$ to a coloring in $P_\cert$ by letting $\chi_\cert(p_{ij}) = \chi(p_i)$, for $i \in [n]$ and $j \in [k]$.  
Now we can naturally define the discrepancy induced on $P_\cert$ by any coloring $\chi$ of $P$ as 
$\disc_{\chi_\cert} (P_\cert, \Eu{A}) = \max_{r \in \Eu{A}} \sum_{p_{i,j} \in P_\cert \cap r} \chi(p_{i,j})$.

\Paragraph{From low-discrepancy to $\eps$-samples}
There is a well-studied relationship between range spaces that admit low-discrepancy colorings, and creating $\eps$-samples of those range spaces~\cite{CM96,Mat99,Cha01,Bec81a}.  The key relationship states that if $\disc(X,\Eu{A}) = \gamma \log^{\omega}(n)$, then there exists an $\eps$-sample of $(P,\Eu{A})$ of size $O((\gamma/\eps) \cdot \log^\omega(\gamma/\eps))$~\cite{Phi08}, for values $\gamma, \omega$ independent of $n$ or $\eps$.  
Construct the coloring, and with equal probability discard either all points colored 
either $-1$ or those colored $+1$.  This roughly halves the point set size, and also implies zero over-count in expectation for any fixed range.  Repeat this coloring and reduction of points until the desired size is achieved.  This can be done efficiently in a distributed manner through a merge-reduce framework~\cite{CM96};   
The take-away is that a method for a low-discrepancy coloring directly implies a method to create an $\eps$-sample, where the counting error is in expectation zero for any fixed range.  
We describe and extend these results in much more detail in Appendix \ref{app:MR}.  

\section{\RE Coresets}\label{sec:rec}
First we will analyze $\eps$-\RE coresets through the $P_\cert$ interpretation of uncertain point set $P$.  The canonical transversals $P_j$ of $P$ will also be useful.  In Section \ref{sec:RE-disc} we will relate these results to a form of discrepancy.  

\begin{lemma}\label{lem:certtocomp}
 $T \subset P$ is an $\eps$-\RE coreset for $(P,\Eu{A})$ if and only if $T_{\text{cert}} \subset P_\cert$ is an $\eps$-sample for $(P_\cert,\Eu{A})$.  
\end{lemma}
\begin{proof}
First note that since $\Pr[p_i = p_{ij}] = \frac{1}{k}$ $\forall i,j$, hence by linearity of expectations we have that 
$\E_{Q \Subset P} [|Q \cap r|] = \sum_{i=1}^n E[|p_i \cap r|] =\frac{1}{k} |P_\cert \cap r |$. 
Now, direct computation gives us:
\begin{align*}
\left| \frac{|P_\cert \cap r|}{|P_\cert|} - \frac{|T_\cert \cap r|}{|T_\cert|} \right|
= \left| \frac{|P_\cert \cap r|}{k |P|} - \frac{|T_\cert \cap r|}{k |T|} \right| 
= \left|  E_{r(P)} - E_{r(T)} \right|  < \eps. 
\end{align*}
\end{proof}

The next implication enables us to determine an $\eps$-\RE coreset on $P$ from $\eps$-samples on each $P_j \Subset P$.  Recall $P_j$ is the $j$th canonical transversal of $P$ for $j \in [k]$, and is defined similarly for a subset $T \subset P$ as $T_j$.

\begin{lemma}\label{lem:eachcomp}
Given a range space $(P_\cert,\Eu{A})$, if we have $T \subset P$ such that $T_j$ is an $\eps$-sample for $(P_j, \Eu{A})$ for all $j \in [k]$, then $T$ is an $\eps$-\RE coreset for $(P, \Eu{A})$. 
\end{lemma}

\begin{proof} Consider an arbitrary range $r \in \Eu{R}$, and compute directly  
 $\left| E_{r(P)} - E_{r(T)}  \right| $.
Recalling that $E_{r(P)} = \frac{|P_\cert \cap r|}{|P_\cert|}$
and observing that $|P_\cert| = k |P|$, we get that:
\begin{align*}
\left| E_{r(P)} - E_{r(T)} \right| 
= 
\left| \frac{\sum_{j=1}^k |P_j \cap r|}{k |P|} -  \frac{\sum_{j=1}^k |T_j \cap r|}{k |T|} \right|  
\leq 
\frac{1}{k}  \sum_{j=1}^k \left| \frac{|P_j \cap r|}{|P|} - \frac{|T_j \cap r|}{|T|} \right| 
\leq 
\frac{1}{k} (k \eps) 
= \eps. 
\end{align*}
\end{proof}

\subsection{Random Sampling}
 We show that a simple random sampling gives us an $\eps$-\RE coreset of $P$.
  
\begin{theorem}
For an uncertain points set $P$ and range space $(P_\cert, \Eu{A})$ with VC-dimension $\nu$,  a random sample $T \subset P$ of size $O((1/\eps^2) (\nu + \log (k/\delta)))$ is an $\eps$-\RE coreset of $(P,\Eu{I}$) with probability at least $1- \delta$.
\label{thm:RE-samp}
\end{theorem}

\begin{proof} 
A random sample $T_j$ of size $O((1/\eps^2) (\nu + \log (1/\delta')))$ is an $\eps$-sample of any $(P_j, \Eu{A})$ with probability at least $1-\delta'$ ~\cite{LLS01}.  Now assuming $T \subset P$ resulted from a random sample on $P$, it induces the $k$ disjoint canonical transversals $T_j$ on $T$, such that $T_j \subset P_j$ and $|T_j| = O((1/\eps^2) (\nu + \log (1/\delta')))$ for $j \in [k]$.  
Each $T_j$ is an $\eps$-sample of $(P_j,\Eu{A})$ for any single $j \in [k]$ with probability at least $1 - \delta'$. 
Following Lemma \ref{lem:eachcomp} and using union bound, we conclude that $T \subset P$ is an $\eps$-\RE coreset for uncertain point set $P$ with probability at least $1 - k\delta'$.  Setting $\delta' = \delta/k$ proves the theorem.  
\end{proof}

\subsection{\RE-Discrepancy and its Properties}
\label{sec:RE-disc}
Next we extend the well-studied relationship between geometric discrepancy and $\eps$-samples on certain data towards $\eps$-\RE coresets on uncertain data.  

We first require precise and slightly non-standard definitions.
 
We introduce a new type of discrepancy based on the expected value of uncertain points called \emph{\RE-discrepancy}.  Let $P_\chi^+$ and $P_\chi^-$ denote the sets of uncertain points from $P$ colored $+1$ or $-1$, respectively, by $\chi$.  
Then $\REd_\chi(P,r) =  |P| \cdot |E_{r(P_\chi^+)} - E_{r(P)}|$ for any $r \in \Eu{A}$. 
The usual extensions then follow:
$\REd_\chi(P,\Eu{A}) = \max_{r \in \Eu{A}} \REd(P,r)$ and
$\REd(P,\Eu{A}) = \min_\chi \REd_\chi(P,\Eu{A})$.  
Note that $(P,\Eu{A})$ is technically not a range space, since $\Eu{A}$ defines subsets of $P_\cert$ in this case, not of $P$.

\begin{lemma}\label{lem:REcolor}
Consider a coloring $\chi : P \to \{-1,+1\}$ such that $\REd_\chi(P,\Eu{A}) = \gamma \log^\omega(n)$ and $|P_\chi^+| = n/2$.  
Then the set $P_\chi^+$ is an $\eps$-\RE coreset of $(P,\Eu{A})$ with $\eps = \frac{ \gamma}{n} \log (n)$.

Furthermore, if a subset $T \subset P$ has size $n/2$ and is an $(\frac{ \gamma}{n} \log^\omega (n))$-\RE coreset, then it defines a coloring $\chi$ (where $\chi(p_i) = +1$ for $p_i \in T$) that has $\REd_\chi(P,\Eu{A}) = \gamma \log^\omega(n)$.  
\end{lemma}

\begin{proof}
We prove the second statement, the first follows symmetrically.  
We refer to the subset $T$ as $P_\chi^+$.  
Let  $r = \arg \max_{r' \in \Eu{A}} |E_{r'(P)} - E_{r'(P_\chi^+)}|$.  This implies  
$
\frac{\gamma}{n} \log^\omega n
\geq 
| E_{r(P)} - E_{r(P_\chi^+)} | 
= 
%\left| \E\left[ \frac{|P \cap r|}{|P|}\right] - \E\left[ \frac{|P_\chi^+ \cap r|}{|P_\chi^+|} \right] \right|
%= 
\frac{1}{n} \REd_\chi(P,r). %\qedhere
$
\end{proof}

We can now recast \RE-discrepancy to discrepancy on $P_\cert$. From  Lemma \ref{lem:certtocomp} 
$\left| \frac{|P_\cert \cap r|}{k |P|} - \frac{|T_\cert \cap r|}{k |T|} \right| =\left|  E_{r(P)} - E_{r(T)} \right|$ and
after some basic substitutions we obtain the following.  

\begin{lemma}\label{lem:pcerttoexp}
$\REd_\chi(P,\Eu{A}) = 
\frac{1}{k} \disc_{\chi_{\cert}} (P_\cert , \Eu{A})$.
\end{lemma}

This does not immediately solve $\eps$-\RE coresets by standard discrepancy  techniques on $P_\cert$ because we need to find a coloring $\chi$ on $P$.  A coloring $\chi_\cert$ on $P_\cert$ may not be consistent across all $p_{i,j} \in p_i$.
The following lemma allows us to reduce this to a problem of coloring each canonical transversal $P_j$.

\begin{lemma}\label{lem:jointdisc}
$\REd_{\chi}(P,\Eu{A}) \leq  \max_j \disc_{\chi_\cert}(P_j, \Eu{A}).$
\end{lemma} 
\begin{proof}
For any $r \in \Eu{A}$ and any coloring $\chi$ (and the corresponding $\chi_\cert$), we can write $P$ as a union of disjoint transversals $P_j$ to obtain 
\begin{align*}
 \disc_{\chi_\cert}(P_\cert,r) 
 &=
\bigg|\sum_{j=1}^k \sum_{p_{ij} \in P_j \cap r} \chi_\cert(p_{ij}) \bigg|  
\leq 
\sum_{j=1}^k \bigg| \sum_{p_{ij} \in P_j \cap r} \chi_\cert(p_{ij}) \bigg| 
\\ & \leq
 \sum_{j=1}^k \disc_{\chi_\cert} (P_j, r) 
\leq
k \max_j \disc_{\chi_\cert}(P_j,r).	  
\vspace{-3mm}
\end{align*}
Since this holds for every $r \in \Eu{A}$, hence (using Lemma \ref{lem:pcerttoexp})
\[
\REd_\chi(P,\Eu{A})  = \frac{1}{k} \disc_{\chi_\cert}(P_\cert,\Eu{A}) \leq \max_j \disc_{\chi_\cert}(P_j, \Eu{A}). \qedhere
\]
\end{proof}

\subsection{$\eps$-\RE Coresets in $\bl{R}^1$}
\label{subsec:1dre}
\begin{lemma}\label{lowdis}
Consider uncertain point set $P$ with $P_\cert \subset \bl{R}^1$ and the range space $(P_\cert, \Eu{I}_+)$ with ranges defined by one-sided intervals of the form $(-\infty, x]$, then $\REd(P,\Eu{I}) = O(\sqrt{k} \log n)$.
 \end{lemma}
\begin{proof}
 Spencer et. al.~\cite{spencer} show that $\disc((P, \Sigma_k),\Eu{I}_k)$ is $O(\sqrt{k} \log n)$.   Since we obtain the $\Sigma_k$ from the canonical transversals $P_1$ through $P_k$, by definition this results in upper bounds on the the discrepancy over all $P_j$ (it bounds the max). Lemma \ref{lem:jointdisc} then gives us the bound on $\REd(P,\Eu{I})$.
\end{proof} 

As we discussed in Appendix \ref{app:MR} the low \RE-discrepancy coloring can be iterated in a merge-reduce framework as developed by Chazelle and Matousek~\cite{CM96}.   With Theorem \ref{thm:RE-disc2samp} we can prove the following theorem.  

\begin{theorem}\label{thm:1deps-RE}
Consider uncertain point set $P$ and range space $(P_\cert, \Eu{I}_+)$ with ranges defined by one-sided intervals of the form $(-\infty, x]$, 
then an $\eps$-\RE coreset can be constructed of size $O((\sqrt{k}/\eps) \log (k/\eps))$.
 \end{theorem}

Since expected value is linear, \RE-$\disc_\chi(P,(-\infty,x]) -$ \RE-$\disc_\chi(P,(-\infty,y)) = $ \RE-$\disc_\chi(P,[y,x])$ for $y<x$ and the above result also holds for the family of two-sided ranges $\Eu{I}$.

\subsection{$\eps$-\RE Coresets for Rectangles in $\bl{R}^d$}
\label{subsec:highdre}
Here let $P$ be a set of $n$ uncertain points where each possible location of a point $p_{i,j} \in \bl{R}^d$.  We consider a range space $(P_\cert, \Eu{R}_d)$ defined by $d$-dimensional axis-aligned rectangles.  

Each canonical transversal $P_j$ for $j \in [k]$ no longer implies a unique permutation on the points (for $d>1$).  But, for any rectangle $r \in \Eu{R}$, we can represent any $r \cap P_j$  as the disjoint union of points $P_j$ contained in intervals on a predefined set of ${(1 + \log n )}^{d-1}$ permutations~\cite{bohus}. 
Spencer \etal~\cite{spencer} showed there exists a coloring $\chi$ such that 
\[
\max_j \disc_{\chi}(P_j, \Eu{R}) = O(D_\ell(n) \log ^{d-1} n),
\] 
where $\ell ={(1 + \log n )}^{d-1}$ is the number of defined permutations and $D_\ell(n)$ is the discrepancy of $\ell$ permutations over $n$ points and ranges defined as intervals on each permutation.  Furthermore, they showed $D_\ell(n) = O(\sqrt{\ell} \log n)$.

To get the $\RE$-discrepancy bound for $P_{\cert} = \cup_{j=1}^k P_j$, we first decompose $P_{\cert}$ into the $k$ point sets $P_j$ of size $n$.
We then obtain ${(1 + \log n )}^{d-1}$ permutations over points in each $P_j$, and hence obtain a family $\Sigma_\ell$ of $\ell = {k (1 + \log n )}^{d-1}$ permutations over all $P_j$. $D_\ell(n) = O(\sqrt{\ell} \log n)$ yields 
\[
\disc((P, \Sigma_\ell), \Eu{I}_\ell) =  O(\sqrt{k} \log^{\frac{d+1}{2}} n).
\]
Now each set $P_j \cap r$ for $r \in \Eu{R}_d$, can be written as the disjoint union of $O{( \log^{d-1} n)}$ intervals of $\Sigma_\ell$.  Summing up over each interval, we get that 
$\disc(P_j, \Eu{R}) = O(\sqrt{k} \log^{\frac{3d-1}{2}} n)$ for each $j$.  
By Lemma \ref{lem:jointdisc} this bounds the $\RE$-discrepancy as well.
Finally, we can again apply the merge-reduce framework of Chazelle and Matousek~\cite{CM96} (via Theorem \ref{thm:RE-disc2samp}) to achieve an $\eps$-\RE coreset.

\begin{theorem}\label{thm:highdeps-RE}
Consider uncertain point set $P$ and range space $(P_\cert, \Eu{R}_d)$ (for $d>1$) with ranges defined by axis-aligned rectangles in $\bl{R}^d$. 
Then an $\eps$-\RE coreset can be constructed of size $O((\sqrt{k}/\eps) \log^{\frac{3d-1}{2}} (k/\eps))$.
 \end{theorem}

\section{\RC Coresets}
\label{sec:RC}

Recall that an $\eps$-\RC coreset $T$ of a set $P$ of $n$ uncertain points satisfies that for all queries $r \in \Eu{A}$ and all thresholds $\tau \in [0,1]$ we have $|G_{P,r}(\tau) - G_{T,r}(\tau)| \leq \eps$, where $G_{P,r}(\tau)$ represents the fraction of points from $P$ that are in range $r$ with probability at least $\tau$.  

In this setting, given a range $r \in \Eu{A}$ and a threshold $\tau \in [0,1]$ we can let the pair $(r, \tau) \in \Eu{A} \times [0,1]$ define a range $R_{r,\tau}$ such that each $p_i \in P$ is either in or not in $R_{r,\tau}$.  Let $(P, \Eu{A} \times [0,1])$ denote this range space.   If $(P_\cert, \Eu{A})$ has VC-dimension $\nu$, then $(P, \Eu{A} \times [0,1])$ has VC-dimension $O(\nu +1)$;  see Corollary 5.23 in \cite{peled}.  This implies that random sampling works to construct $\eps$-RC coresets.  

\begin{theorem}
\label{thm:RC-sample}
For uncertain point set $P$ and range space $(P_\cert, \Eu{A})$ with VC-dimension $\nu$, a random sample $T \subset P$ of size $O((1/\eps^2) (\nu + \log(1/\delta)))$ is an $\eps$-RC coreset of $(P,\Eu{A})$ with probability at least $1-\delta$.  
\end{theorem}

Yang \etal propose a similar result~\cite{aggregate} as above, without proof.

\subsection{\RC Coresets in $\bl{R}^1$}
Constructing $\eps$-\RC coresets when the family of ranges $\Eu{I}_+$ represents one-sided, one-dimensional intervals is much easier than other cases.  It relies heavily on the ordered structure of the canonical permutations, and thus discrepancy results do not need to decompose and then re-compose the ranges.

\begin{lemma} \label{lem:rc-query}
A point $p_i \in P$ is in range $r \in \Eu{I}_+$ with probability at least $\tau = t/k$ if and only if $ p_{i,t} \in r \cap P_t$.
\end{lemma}

\begin{proof}
By the canonical permutations, since for all $i \in [n]$, we require $p_{i,j} < p_{i,j+1}$, then if $p_{i,t} \in r$, it follows that $p_{i,j} \in r$ for $j \leq t$.  
Similarly if $p_{i,t} \notin r$, then all $p_{i,j} \notin r$ for $j \geq t$.  
\end{proof}

Thus when each canonical permutation is represented upto an error $\eps$ by a coreset $T$, then each threshold $\tau$ is represented within $\eps$.  Hence, as with $\eps$-\RE coresets, we invoke the low-discrepancy coloring of Bohus~\cite{bohus} and Spencer \etal~\cite{spencer}, and then iterate them (invoking Theorem \ref{thm:disc2samp}) to achieve a small size $\eps$-\RC coreset.

\begin{theorem}\label{thm:1deps-RC}
For uncertain point set $P$ and range space $(P_\cert, \Eu{I}_+)$ with ranges defined by one-sided intervals of the form $(-\infty, a]$.  
An $\eps$-\RC coreset of $(P,\Eu{I}_+)$ can be constructed of size $O((\sqrt{k}/\eps) \log (k/\eps))$.
\end{theorem}

Extending Lemma \ref{lem:rc-query} from one-sided intervals of the form $[-\infty, a] \in \Eu{I}_+$ to intervals of the form $[a, b] \in \Eu{I}$ turns out to be non-trivial.  It is \emph{not} true that $G_{P,[a,b]}(\tau) = G_{P,[-\infty,b]}(\tau) - G_{P,[-\infty,a]}(\tau)$, hence the two queries cannot simply be subtracted.  Also, while the set of points corresponding to the
query $G_{P,[-\infty, a]}(\frac{t}{k})$ are a contiguous interval in the $t$th permutation we construct in Lemma \ref{lem:rc-query}, the same need not be true of points corresponding to $G_{P,[a,b]}(\frac{t}{k})$.
This is a similar difficulty in spirit as noted by Kaplan \etal~\cite{colors} in the problem of counting the number of points of distinct colors in a box where one cannot take a naive decomposition and add up the numbers returned by each subproblem. 

We give now a construction to solve this two-sided problem for uncertain points in $\bl{R}^1$ inspired by that of Kaplan \etal~\cite{colors}, but we require specifying a fixed value of $t \in [k]$.  Given an uncertain point $p_i \in P$ assume w.l.o.g that $p_{i,j} < p_{i,j+1}$.  Also pretend there is a point $p_{i,k+1} = \eta$ where $\eta$ is larger than any $b \in \bl{R}^1$ from a query range $[a,b]$ (essentially $\eta = \infty$).  
Given a range $[a,b]$, we consider the right-most set of $t$ locations of $p_i$ (here $\{p_{i,j-t}, \ldots, p_{i,j}\}$) that are in the range.  This satisfies 
(i) $p_{i,j - t} \geq a$, 
(ii) $p_{i,j} \leq b$, and 
(iii) to ensure that it is the right-most such set, $p_{i,j+1} > b$.  

To satisfy these three constraints we re-pose the problem in $\bl{R}^3$ to designate each contiguous set of $t$ possible locations of $p_i$ as a single point.  So for $t < j \leq k$, we map $p_{i,j}$ to $\bar{p}_{i,j}^t = (p_{i,j-t}, p_{i,j}, p_{i,j+1})$.  
Correspondingly, a range $r = [a,b]$ is mapped to a range $\bar{r} = [a, \infty) \times (-\infty,b] \times (b,\infty)$;  
see Figure \ref{fig:queryreduction}.
Let $\bar p_i^t$ denote the set of all $\bar p^t_{i,j}$, and let $\bar{P}^t$ represent $\bigcup_i \bar{p}_i^t$.  
 
\begin{figure}[t]
\centering
\includegraphics[width=0.8\linewidth]{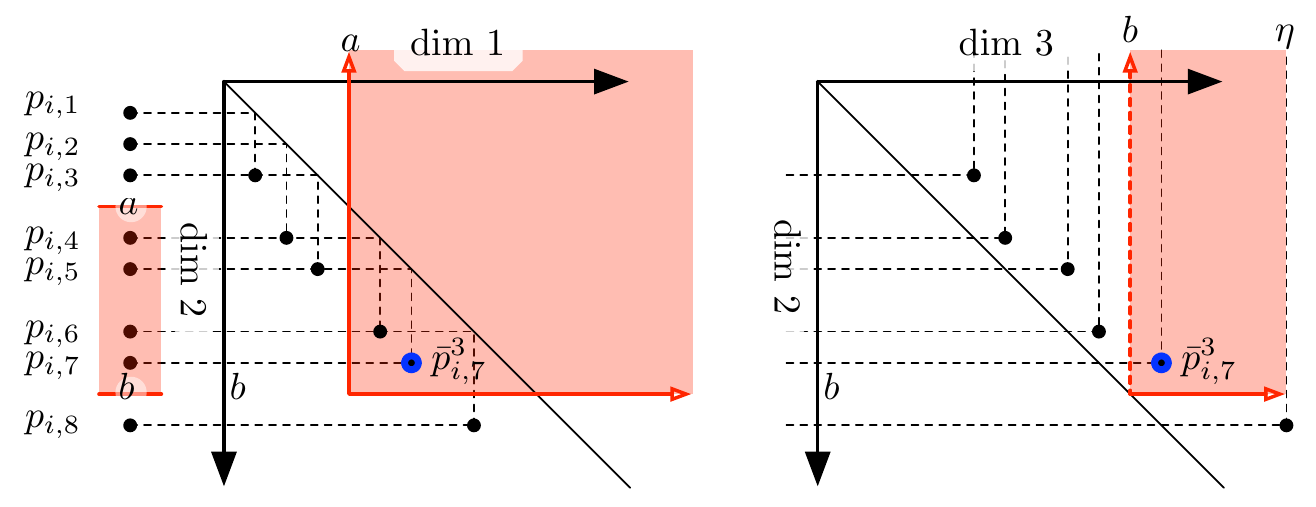}
\vspace{-4mm}
\caption{\small Uncertain point $p_i$ queried by range $[a,b]$.  Lifting shown to $\bar{p}_i^3$ along dimensions 1 and 2 (left) and along dimensions 2 and 3 (right).
\vspace{-4mm}
}  

\label{fig:queryreduction}
\end{figure} 
 
\begin{lemma}
$p_i$ is in interval $r = [a,b]$ with threshold at least $t/k$ if and only if $\bar{p}^t_i \cap \bar{r}^t \geq 1$.
Furthermore, no two points $p_{i,j}, p_{i,j'} \in p_i$ can map to points $\bar{p}^t_{i,j}, \bar{p}^t_{i,j'}$ such that both are in a range $\bar{r}^t$.  
\end{lemma}
 
\begin{proof}
Since $p_{i,j} < p_{i,j+1}$, then if $p_{i,j-t} \geq a$ it implies all $p_{i,\ell} \geq a$ for $\ell \geq j-t$, and similarly, if $p_{i,j} \leq b$ then all $p_{i,\ell} \leq b$ for all $\ell \leq j$.  Hence if $\bar{p}^t_{i,j}$ satisfies the first two dimensional constraints of the range $\bar{r}^t$, it implies $t$ points $p_{i,j-t} \ldots, p_{i,j}$ are in the range $[a,b]$.  
Satisfying the constraint of $\bar{r}^t$ in the third coordinate indicates that $p_{i,j+1} \notin [a,b]$.  There can only be one point $p_{i,j}$ which satisfies the constraint of the last two coordinates that $p_{i,j} \leq b < p_{i,j+1}$.  And for any range which contains at least $t$ possible locations, there must be at least one such set (and only one) of $t$ consecutive points which has this satisfying $p_{i,j}$.  
\end{proof}

\begin{corollary}
Any uncertain point set $P \in \bl{R}^1$ of size $n$ and range $r = [a,b]$ has $G_{P, r}(\frac{t}{k}) = |\bar{P}^t \cap \bar{r}^t|/n$.  
\end{corollary}

This presents an alternative view of each uncertain point in $\bl{R}^1$ with $k$ possible locations as an uncertain point in $\bl{R}^3$ with $k-t$ possible locations (since for now we only consider a threshold $\tau = t/k$).  Where $\Eu{I}$ represents the family of ranges defined by two-sided intervals, let $\bar{\Eu{I}}$ be the corresponding family of ranges in $\bl{R}^3$ of the form $[a,\infty) \times (-\infty,b] \times (b, \infty)$ corresponding to an interval $[a,b] \in \Eu{I}$.  Under the assumption (valid under the lifting defined above) that each uncertain point can have at most one location fall in each range, we can now decompose the ranges and count the number of points that fall in each sub-range and add them together.  
Using the techniques (described in detail in Section \ref{subsec:highdre}) of Bohus~\cite{bohus} and Spencer \etal~\cite{spencer} we can consider $\ell =(k-t) (1+\lceil \log n \rceil)^2$ permutations of $\bar{P}_\cert^t$ such that each range $\bar{r} \in \bar{\Eu{I}}$ can be written as the points in a disjoint union of intervals from these permutations.  
To extend low discrepancy to \emph{each} of the $k$ distinct values of threshold $t$, there are $k$ such liftings and $h = k\cdot \ell = O(k^2 \log^2 n)$ such permutations we need to consider.   We can construct a coloring $\chi : P \to \{-1,+1\}$ such that intervals on each permutation has discrepancy $O(\sqrt{h} \log n) = O(k \log^2 n )$. Recall that for any fixed threshold $t$ we only need to consider the corresponding $\ell$ permutations, hence the total discrepancy for any such range is at most the sum of discrepancy from all corresponding 
$\ell=O( k \log^2 n)$ permutations or $O(k^2 \log^4 n)$.  
Finally, this low-discrepancy coloring can be iterated (via Theorem \ref{thm:disc2samp}) to achieve the following theorem.  

\begin{theorem}\label{thm:rcsample1d}
Consider an uncertain point set $P$ along with ranges $\Eu{I}$ of two-sided intervals.  We can construct an $\eps$-RC coreset $T$ for $(P,\Eu{I})$ of size $O((k^2 /\eps) \log^4 (k/\eps))$.  
\end{theorem}

\subsection{\RC Coresets for Rectangles in $\bl{R}^d$}

The approach for $\Eu{I}$ can be further extended to $\Eu{R}_d$, axis-aligned rectangles in $\bl{R}^d$. Again the key idea
is to define a proxy point set $\bar{P}$ such that $| \bar r \cap \bar{P}|$ equals the number of uncertain points in $r$ with
at least threshold $t$. This requires a suitable lifting map and decomposition of space to prevent over or under counting;
we employ techniques from Kaplan \etal~\cite{colors}.

First we transform queries on axis-aligned rectangles in $\bl{R}^d$ to the semi-bounded case in $\bl{R}^{2d}$.  Denote the $x_i$-coordinate of a point $q$ as $x_i(q)$, we double all the coordinates of each point $q = (x_1(q),...,x_{\ell}(q),..., x_d(q))$ to obtain point \[ \tilde{q} = (-x_1(q),x_1(q)...,-x_{\ell}(q),x_{\ell}(q),..., -x_d(q),x_d(q))\] in $\bl{R}^{2d}$. 
Now answering range counting query $\prod_{i=1}^d [a_i, b_i]$ is equivalent to solving the query  \[
\prod_{i=1}^d\left[(-\infty, -a_i] \times (-\infty, b_i] \right]
\]
on the lifted point set.

Based on this reduction we can focus on queries of \emph{negative orthants} of the form $\prod_{i=1}^d (-\infty, a_i]$ and represent each orthant by its apex $a =(a_1, ..., a_{d}) \in \bl{R}^{d}$ as $Q_a^-$.  Similarly, we can define $Q_a^+$ as \emph{positive orthants} in the form $\prod_{i=1}^d [a_i, \infty) \subseteq \bl{R}^d$. For any point set $A \subset \bl{R}^d$ define $U(A) = \cup_{a \in A} Q_a^+$.

A \emph{tight} orthant has \emph{a} location of $p_i \in P$ incident to every bounding facet.    
Let $C_{i,t}$ be the set of all apexes representing tight negative orthants that contain exactly $t$ locations of $p_i$; see Figure \ref{fig:rc123}(a).   An important observation is that query orthant $Q_a^-$ contains $p_i$ with threshold at least $t$ if and only if it contains at least one point from $C_{i,t}$.
 
Let $Q^+_{i,t} = \cup_{c \in C_{i,t} }Q_c^+$ be the locus of all negative orthant query apexes that contain at least $t$ locations of $p_i$; see Figure \ref{fig:rc123}(b).  
Notice that $Q^+_{i,t} = U(C_{i,t})$.
 
\begin{figure}[t]
\centering
\includegraphics[width=0.3\linewidth]{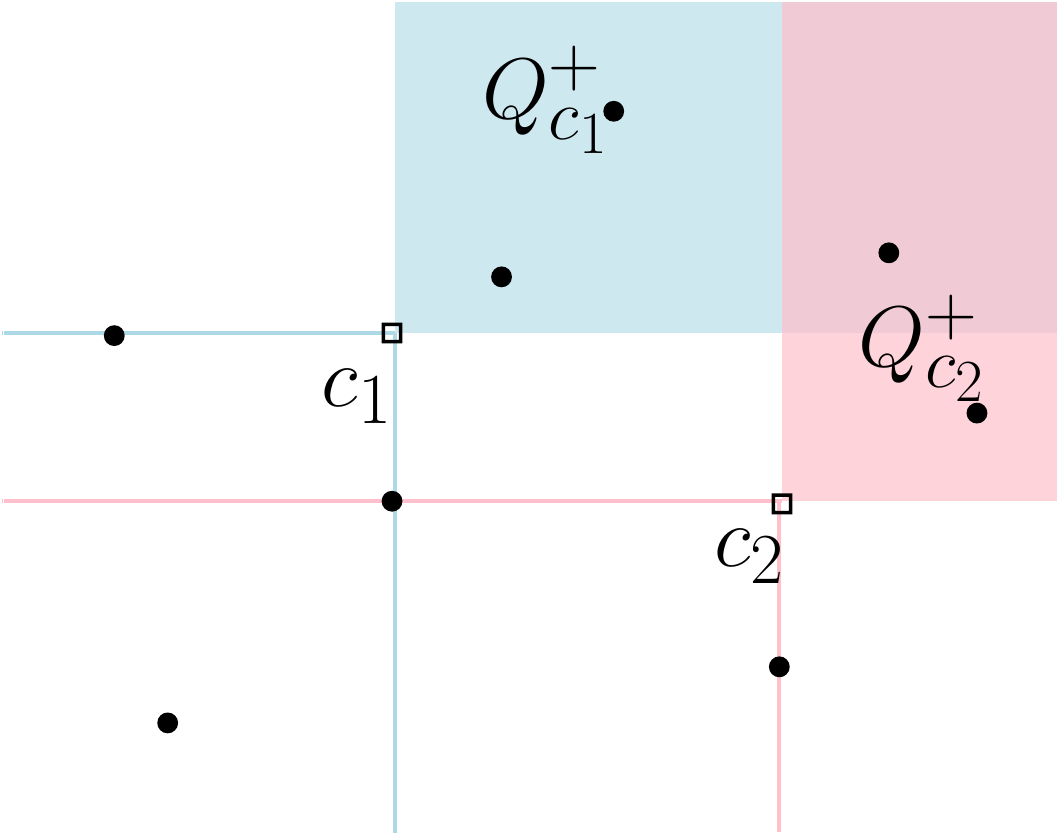}
\hspace{2mm}
\includegraphics[width=0.3\linewidth]{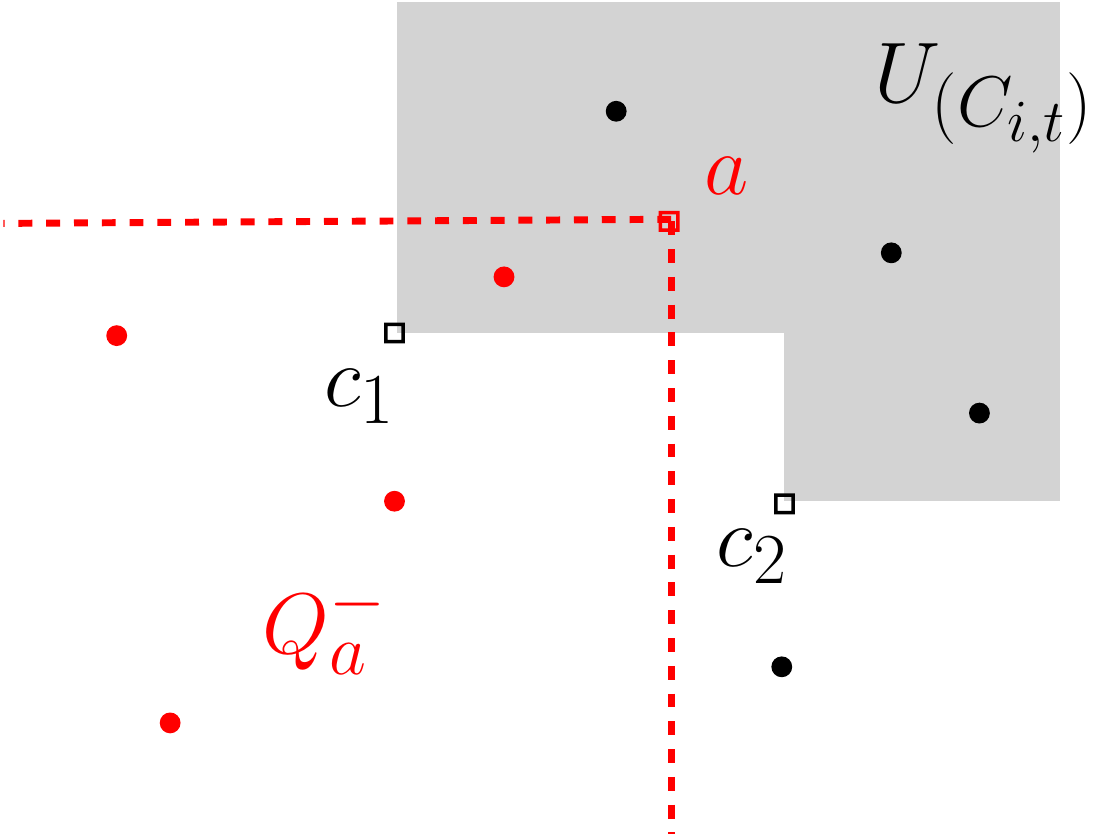}
\hspace{2mm}
\includegraphics[width=0.3\linewidth]{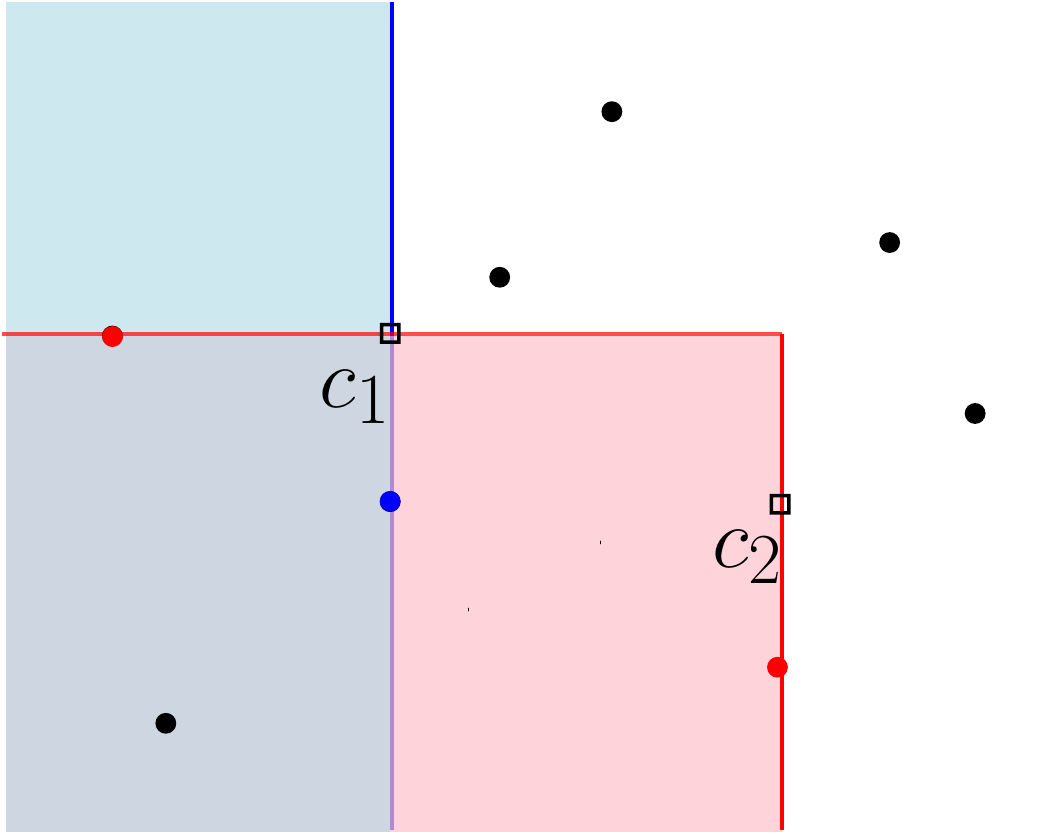}

\vspace{-1mm}
(a)  \hspace{1.85in} (b) \hspace{1.85in} (c)

\vspace{-3mm}
\caption{\small Illustration of uncertain point $p_i \in \bl{R}^2$ with $k = 8$ and $t = 3$. 
(a) All tight negative orthants containing exactly $t=3$ locations of $p_i$, their apexes are $C_{i,t} = \{c_1, c_2\}$. 
(b): $U(C_{i,t})$ is shaded and query $Q_a^-$.
(c): $M(C_{i,t})$, the maximal negative orthants of $C_{i,t}$ that are also bounded in the $x$-direction.  
\vspace{-4mm}}

\label{fig:rc123}
\end{figure}

\begin{lemma}
For any point set $p_i \subset \bl{R}^d$ of $k$ points and some threshold $1 \leq t \leq k$, we can decompose $U(C_{i,t})$ into $f(k) = O(k^d)$ pairwise disjoint boxes, $B(C_{i,t})$.
\end{lemma}
 
\begin{proof}
Let $M(A)$ be the set of maximal empty negative orthants for a point set $A$, such that any $m \in M(A)$ is also bounded in the positive direction 
along the $1$st coordinate axis.  
Kaplan \etal~\cite{colors} show (within Lemma 3.1) that $|M(A)| = |B(A)|$ and provide a specific construction of the boxes $B$.  Thus we only need to bound $|M(C_{i,t})|$ to complete the proof; see $M(C_{i,t})$ in Figure \ref{fig:rc123}(c).  
We note that each coordinate of each $c \in C_{i,t}$ must be the same as some $p_{i,j} \in p_i$.  Thus for each coordinate, among all $c \in C_{i,t}$ there are at most $k$ values.  And each maximal empty tight orthant $m \in M(C_{i,t})$ is uniquely defined by the $d$ coordinates along the axis direction each facet is orthogonal to.  Thus $|M(C_{i,t})|  \leq k^d$, completing the proof.  
\end{proof}
 
Note that as we are working in a lifted space $\bl{R}^{2d}$, this corresponds to $U(C_{i,t})$ being decomposed into $f(k) = O(k^{2d})$ pairwise \emph{disjoint} boxes in which $d$ is the dimensionality of our original point set.

\begin{lemma}
For negative orthant queries $Q_a^-$ with apex $a$ on uncertain point set $P$, a point $p_i \in P$ is in $Q_a^-$ with probability at least $t/k$ if $a$ is in some box in $B(C_{i,t})$, and $a$ will lie in at most one box from $B(C_{i,t})$.  
 \end{lemma}
\begin{proof}
The query orthant $Q_a^-$ contains point $p_i$ with threshold at least $t$ if and only if $Q_a^-$ contains at least one point from $C_{i,t}$ and this happens only when $a \in U(C_{i,t})$. Since the union of constructed boxes in $B(C_{i,t})$ is equivalent to $U(C_{i,t})$ and they are disjoint, the result follows.
\end{proof} 
 
\begin{corollary}
The number of uncertain points from $P$ in query range $Q_a^-$ with probability at least $t/k$ is exactly the number of boxes in $\cup_i^n B(C_{i,t})$ that contain $a$.  
 \end{corollary}
 
Thus for a set of boxes representing $P$, we need to perform count stabbing queries with apex $a$ and show a low-discrepancy coloring of boxes.  

We do a second lifting by transforming each point $a \in \bl{R}^d$ to a semi-bounded box $\bar{a} = \prod_{i=1}^d ((-\infty, a_i] \times [a_i, \infty))$ and each box $b \in \bl{R}^d$  of the form $ \prod_i^d [x_i, y_i]$ to a point $\bar{b} = (x_1, y_1, ..., x_{\ell}, y_{\ell}, ..., x_d, y_d)$ in $\bl{R}^{2d}$. It is easy to verify that $a \in b$ if and only if $\bar{b} \in \bar{a}$.

Since this is our second doubling of dimension, we are now dealing with points in $\bl{R}^{4d}$. Lifting $P$ to $\bar{P}$ in $\bl{R}^{4d}$ now presents an alternative view of each uncertain point $p_i \in P $ as an uncertain point $\bar{p_i}$ in $\bl{R}^{4d}$ with $g_k = O(k^{2d})$ possible locations with the query boxes represented as $\bar{\Eu{R}}$ in $\bl{R}^{4d}$.

 We now proceed similarly to the proof of Theorem \ref{thm:rcsample1d}. For a fixed threshold $t$, obtain $\ell = g_k \cdot (1+\lceil \log n \rceil)^{4d-1}$ disjoint permutations of $\bar{P}_\cert^t$ such that each range $\bar{r} \in \bar{\Eu{R}}$ can be written as the points in a disjoint union of intervals from these permutations. 
 For the $k$ distinct values of $t$, there are $k$ such liftings and $h = O \left(k \cdot g_k \cdot \log^{4d-1} n \right)$ such permutations we need to consider, and we can construct a coloring $\chi : P \to \{-1,+1\}$ so that intervals on each permutation have discrepancy \[O(\sqrt{h} \log n) = O(k^{d+\frac{1}{2}} \log^{\frac{4d+1}{2}} n ).\] 
  Hence for any such range and specific threshold $t$, the total discrepancy is the sum of discrepancy from all corresponding $\ell = O \left( g_k \cdot \log^{4d-1} n \right)$ permutations, or $O \left(k^{3d +\frac{1}{2}} \log^{6d - \frac{1}{2}} n \right)$.  
By applying the iterated low-discrepancy coloring (Theorem \ref{thm:disc2samp}), we achieve the following result.

\begin{theorem}
\label{thm:RC-Rd}
Consider an uncertain point set $P$ and range space $(P_\cert, \Eu{R}_d)$ with ranges defined by axis-aligned rectangles in $\bl{R}^d$. 
Then an $\eps$-\RC coreset can be constructed of size $O \left((k^{3d +\frac{1}{2}}/\eps) \log^{6d - \frac{1}{2}} (k/\eps) \right)$.
\end{theorem}

\section{\RQ Coresets}

In this section, given an uncertain point set $P$ and its $\eps$-\RE coreset $T$, we want to determine values $\eps'$ and $\alpha$ so $T$ is an $(\eps',\alpha)$-\RQ coreset.  
That is for any $r \in \Eu{A}$ and threshold $\tau \in [0,1]$ there exists a $\gamma \in [\tau-\alpha, \tau+\alpha]$ such that 
\[
\Big| \Pr_{Q \Subset P} \hspace{-1mm} \left[\frac{|Q \cap r|}{|Q|} \leq \tau \right]  -
\Pr_{S \Subset T} \hspace{-1mm} \left[\frac{|S \cap r|}{|S|} \leq \gamma \right] \Big| \leq \eps'.
\]

At a high level, our tack will be to realize that both $|Q \cap r|$ and $|S \cap r|$ behave like Binomial random variables.  By $T$ being an $\eps$-\RE coreset of $P$, then after normalizing, its mean is at most $\eps$-far from that of $P$.  Furthermore, Binomial random variables tend to concentrate around their mean--and more so for those with more trials.  This allows us to say $|S \cap r|/|S|$ is either $\alpha$-close to the expected value of $|Q \cap r|/|Q|$ or is $\eps'$-close to $0$ or $1$.  Since $|Q \cap r|/|Q|$ has the same behavior, but with more concentration, we can bound their distance by the $\alpha$ and $\eps'$ bounds noted before.  We now work out the details.  
  
\begin{theorem}
If $T$ is an $\eps$-\RE coreset of $P$ for $\eps \in (0,1/2)$, 
then  $T$ is an $(\eps',\alpha)$-\RQ coreset for $P$ for $\eps',\alpha \in (0,1/2)$ and satisfying
$\alpha \geq \eps +  \sqrt{(1/2|T|) \ln(2/\eps')}$.  
\label{thm:RQ-CLM}
\end{theorem}
\begin{proof}
We start by examining a Chernoff-Hoeffding bound on a set of independent random variables  $X_i$ so that each $X_i \in [a_i, b_i]$ with $\Delta_i = b_i - a_i$.  Then for some parameter $\beta \in (0, \sum_i \Delta_i/2)$ 
\[
\Pr \left[ \left| \sum_i X_i - \E\left[\sum_i X_i\right] \right| \geq \beta \right] 
\leq 
2 \exp \left(\frac{-2\beta^2}{\sum_i \Delta_i^2}\right).
\]

Consider any $r \in \Eu{A}$.  
We now identify each random variable $X_i = 1 (q_i \in r)$ (that is, $1$ if $q_i \in r$ and $0$ otherwise) where $q_i$ is the random instantiation of some $p_i \in T$.  So $X_i \in \{0,1\}$ and $\Delta_i = 1$.  
Thus by equating $|S \cap r| = \sum X_i$
\begin{align*}
\Pr_{S \Subset T} \left[ \left| |S \cap r| - \E\left[|S \cap r|\right] \right| \geq \beta |S| \right]  \leq 
2 \exp \left(\frac{-2\beta^2 |S|^2}{\sum_i \Delta_i^2}\right)
 = 2 \exp(-2 \beta^2 |S|) \leq \eps'.
\end{align*}
Thus by solving for $\beta$ (and equating $|S| = |T|$)
\[
\Pr_{S \Subset T} \left[ \left| \frac{|S \cap r|}{|S|}  - \E\left[\frac{|S \cap r|}{|S|}\right] \right| \geq \sqrt{\frac{1}{2|T|} \ln( \frac{2}{\eps'})}\right] \leq \eps'.
\]
Now by $T$ being an $\eps$-\RE coreset of $P$ then 
\[
\left| \E_{S \Subset T}\left[\frac{|S \cap r|}{|S|}\right] - \E_{Q \Subset P}\left[\frac{|Q \cap r|}{|Q|} \right] \right| \leq \eps.  
\]
Combining these two we have
\[
\Pr_{S \Subset T}  \hspace{-1mm} \left[ \left| \frac{|S \cap r|}{|S|}  - \E_{Q \Subset P} \hspace{-1mm} \left[\frac{|Q \cap r|}{|Q|}\right] \right| \geq \alpha \right] \leq \eps'
\]
for $\alpha = \eps + \sqrt{\frac{1}{2|T|} \ln( \frac{2}{\eps'})}$.

Combining these statements, for any $x \leq M-\alpha \leq M - \alpha'$ 
we have 
$\eps' > F_{T,r}(x) \geq 0$ and $\eps' > F_{P,r}(x) \geq 0$ (and symmetrically for $x \geq M+\alpha \geq M+\alpha')$.  
It follows that $F_{T,r}$ is an $(\eps',\alpha)$-quantization of $F_{P,r}$.  

Since this holds for any $r \in \Eu{A}$, by $T$ being an $\eps$-\RE coreset of $P$, it follows that $T$ is also an $(\eps', \alpha)$-\RQ coreset of $P$.  
\end{proof}

We can now combine this result with specific results for $\eps$-\RE coresets to get size bounds for $(\eps,\alpha)$-\RQ coresets.  To achieve the below bounds we set $\eps = \eps'$.

\begin{corollary}
For uncertain point set $P$ with range space $(P_\cert, \Eu{A})$, there exists a $(\eps, \eps + \sqrt{(1/2|T|) \ln(2/\eps)})$-RQ coreset of $(P,\Eu{A})$ of size $|T| = $
\begin{itemize} \denselist
\item $O((1/\eps^2) (\nu + \log(k/\delta)))$ when $\Eu{A}$ has VC-dimension $\nu$, with probability $1-\delta$ (Theorem \ref{thm:RE-samp}),
\item $O((\sqrt{k} /\eps) \log (k/\eps))$ when $\Eu{A} = \Eu{I}$ (Theorem \ref{thm:1deps-RE}), and
\item $O\left((\sqrt{k}/\eps) \log^{\frac{3d - 1}{2}} (k/\eps) \right)$ when $\Eu{A} = \Eu{R}_d$ (Theorem \ref{thm:highdeps-RE}).
\end{itemize}
\end{corollary}

Finally we discuss why the $\alpha$ term in the $(\eps',\alpha)$-\RQ coreset $T$ is needed.  
Recall from Section \ref{sec:rec} that approximating the value of $\E_{Q \Subset P} \big[ \frac{|Q \cap r|}{|Q|} \big]$  with $\E_{S \Subset T} \big[ \frac{|S \cap r|}{|S|} \big]$ for all $r$ corresponds to a low-discrepancy sample of $P_{\cert}$. Discrepancy error immediately implies we will have at least the $\eps$ horizontal shift between the two distributions and their means, unless we could obtain a zero discrepancy sample of $P_{\cert}$.  Note this $\eps$-horizontal error corresponds to the $\alpha$ term in an $(\eps',\alpha)$-\RQ coreset.  
When $P$ is very large, then due to the central limit theorem, $F_{P,r}$ will grow very sharply around $\E_{Q \Subset P} \big[ \frac{|Q \cap r|}{|Q|} \big]$.  
In the worst case $F_{T,r}$ may be $\Omega(1)$ vertically away from $F_{P,r}$ on either side of $\E_{S \Subset T} \big[ \frac{|S \cap r|}{|S|} \big]$, so no reasonable amount of $\eps'$ vertical tolerance will make up for this gap.  

On the other hand, the $\eps'$ vertical component is necessary since for very small probability events (that is for a fixed range $r$ and small threshold $\tau$) on $P$, we may need a much smaller value of $\tau$ (smaller by $\Omega(1)$) to get the same probability on $T$, requiring a very large horizontal shift.  But since it is a very small probability event, only a small vertical $\eps'$ shift is required.  

The main result of this section then is showing that there exist pairs $(\eps',\alpha)$ which are both small.

\section{Conclusion and Open Questions}
This paper defines and provides the first results for coresets on uncertain data.  These can be essential tools for monitoring a subset of a large noisy data set, as a way to approximately monitor the full uncertainty.  

There are many future directions on this topic, in addition to tightening the provided bounds especially for other range spaces.  Can we remove the dependence on $k$ without random sampling?   Can coresets be constructed over uncertain data for other queries such as minimum enclosing ball, clustering, and extents?  

 \bibliography{uncertain}

\begin{thebibliography}{10}

\bibitem{ACTY09}
{\sc Agarwal, P.~K., Cheng, S.-W., Tao, Y., and Yi, K.}
\newblock Indexing uncertain data.
\newblock In {\em PODS\/} (2009).

\bibitem{AHV07}
{\sc Agarwal, P.~K., Har-Peled, S., and Varadarajan, K.}
\newblock Geometric approximations via coresets.
\newblock {\em Current Trends in Combinatorial and Computational Geometry (E.
  Welzl, ed.)\/} (2007).

\bibitem{AHV04}
{\sc Agarwal, P.~K., Har-Peled, S., and Varadarajan, K.~R.}
\newblock Approximating extent measure of points.
\newblock {\em Journal of ACM 51}, 4 (2004), 2004.

\bibitem{ABSHNSW06}
{\sc Agrawal, P., Benjelloun, O., Sarma, A.~D., Hayworth, C., Nabar, S.,
  Sugihara, T., and Widom, J.}
\newblock Trio: A system for data, uncertainty, and lineage.
\newblock In {\em PODS\/} (2006).

\bibitem{AB99}
{\sc Anthony, M., and Bartlett, P.~L.}
\newblock {\em Neural Network Learning: Theoretical Foundations}.
\newblock Cambridge University Press, 1999.

\bibitem{BHP02}
{\sc B\={a}doiu, M., Har-Peled, S., and Indyk, P.}
\newblock Approximate clustering via core-sets.
\newblock In {\em STOC\/} (2002).

\bibitem{bs-ads-04}
{\sc Bandyopadhyay, D., and Snoeyink, J.}
\newblock Almost-{D}elaunay simplices: Nearest neighbor relations for imprecise
  points.
\newblock In {\em SODA\/} (2004).

\bibitem{Bec81a}
{\sc Beck, J.}
\newblock Roth's estimate of the discrepancy of integer sequences is nearly
  sharp.
\newblock {\em Combinatorica 1\/} (1981), 319--325.

\bibitem{bohus}
{\sc Bohus, G.}
\newblock On the discrepancy of 3 permutations.
\newblock {\em Random Structures and Algorithms 1}, 2 (1990), 215--220.

\bibitem{BC03}
{\sc B\u{a}doiu, M., and Clarkson, K.}
\newblock Smaller core-sets for balls.
\newblock In {\em SODA\/} (2003).

\bibitem{BDJRV05}
{\sc Burdick, D., Deshpande, P.~M., Jayram, T., Ramakrishnan, R., and
  Vaithyanathan, S.}
\newblock {OLAP} over uncertain and imprecise data.
\newblock In {\em VLDB\/} (2005).

\bibitem{Cha01}
{\sc Chazelle, B.}
\newblock {\em The Discrepancy Method}.
\newblock Cambridge, 2000.

\bibitem{CM96}
{\sc Chazelle, B., and Matousek, J.}
\newblock On linear-time deterministic algorithms for optimization problems in
  fixed dimensions.
\newblock {\em Journal of Algorithms 21\/} (1996), 579--597.

\bibitem{CW89}
{\sc Chazelle, B., and Welzl, E.}
\newblock Quasi-optimal range searching in spaces of finite {VC}-dimension.
\newblock {\em Discrete and Computational Geometry 4\/} (1989), 467--489.

\bibitem{threshquery}
{\sc Cheng, R., Xia, Y., Prabhakar, S., Shah, R., and Vitter, J.~S.}
\newblock Efficient indexing methods for probabilistic threshold queries over
  uncertain data.
\newblock In {\em VLDB\/} (2004).

\bibitem{CG09}
{\sc Cormode, G., and Garafalakis, M.}
\newblock Histograms and wavelets of probabilitic data.
\newblock In {\em ICDE\/} (2009).

\bibitem{CLY09}
{\sc Cormode, G., Li, F., and Yi, K.}
\newblock Semantics of ranking queries for probabilistic data and expected
  ranks.
\newblock In {\em ICDE\/} (2009).

\bibitem{efficientquery}
{\sc Dalvi, N., and Suciu, D.}
\newblock Efficient query evaluation on probabilistic databases.
\newblock In {\em VLDB\/} (2004).

\bibitem{gss-egbra-89}
{\sc Guibas, L.~J., Salesin, D., and Stolfi, J.}
\newblock Epsilon geometry: building robust algorithms from imprecise
  computations.
\newblock In {\em SoCG\/} (1989).

\bibitem{gss-cscah-93}
{\sc Guibas, L.~J., Salesin, D., and Stolfi, J.}
\newblock Constructing strongly convex approximate hulls with inaccurate
  primitives.
\newblock {\em Algorithmica 9\/} (1993), 534--560.

\bibitem{peled}
{\sc Har-Peled, S.}
\newblock {\em Geometric Approximation Algorithms}.
\newblock American Mathematical Society, 2011.

\bibitem{hm-ticpps-08}
{\sc Held, M., and Mitchell, J. S.~B.}
\newblock Triangulating input-constrained planar point sets.
\newblock {\em Information Processing Letters 109}, 1 (2008).

\bibitem{JKV07}
{\sc Jayram, T., Kale, S., and Vee, E.}
\newblock Efficient aggregation algorithms for probabilistic data.
\newblock In {\em SODA\/} (2007).

\bibitem{JMMV07}
{\sc Jayram, T., McGregor, A., Muthukrishnan, S., and Vee, E.}
\newblock Estimating statistical aggregates on probabilistic data streams.
\newblock In {\em PODS\/} (2007).

\bibitem{JLP12}
{\sc J{\o}rgensen, A.~G., L\"offler, M., and Phillips, J.~M.}
\newblock Geometric computation on indecisive and uncertain points.
\newblock arXiv:1205.0273.

\bibitem{JLP11}
{\sc J{\o}rgensen, A.~G., L\"offler, M., and Phillips, J.~M.}
\newblock Geometric computation on indecisive points.
\newblock In {\em WADS\/} (2011).

\bibitem{KCS11b}
{\sc Kamousi, P., Chan, T.~M., and Suri, S.}
\newblock The stochastic closest pair problem and nearest neighbor search.
\newblock In {\em WADS\/} (2011).

\bibitem{KCS11a}
{\sc Kamousi, P., Chan, T.~M., and Suri, S.}
\newblock Stochastic minimum spanning trees in euclidean spaces.
\newblock In {\em SOCG\/} (2011).

\bibitem{colors}
{\sc Kaplan, H., Rubin, N., Sharir, M., and Verbin, E.}
\newblock Counting colors in boxes.
\newblock In {\em SODA\/} (2007).

\bibitem{k-bmips-08}
{\sc Kruger, H.}
\newblock Basic measures for imprecise point sets in {$\mathbb{R}^d$}.
\newblock Master's thesis, Utrecht University, 2008.

\bibitem{LLS01}
{\sc Li, Y., Long, P.~M., and Srinivasan, A.}
\newblock Improved bounds on the samples complexity of learning.
\newblock {\em Journal of Computer and System Science 62\/} (2001), 516--527.

\bibitem{LP09}
{\sc L\"offler, M., and Phillips, J.}
\newblock Shape fitting on point sets with probability distributions.
\newblock In {\em ESA\/} (2009), Springer Berlin / Heidelberg.

\bibitem{ls-dtip-08}
{\sc L{\"o}ffler, M., and Snoeyink, J.}
\newblock {Delaunay} triangulations of imprecise points in linear time after
  preprocessing.
\newblock In {\em SOCG\/} (2008).

\bibitem{divide}
{\sc Matousek, J.}
\newblock Approximations and optimal geometric divide-and-conquer.
\newblock {\em Journal of Computer and System Sciences 50}, 2 (1995), 203 --
  208.

\bibitem{Mat99}
{\sc Matou\v{s}ek, J.}
\newblock {\em Geometric Discrepancy}.
\newblock Springer, 1999.

\bibitem{nt-teb-00}
{\sc Nagai, T., and Tokura, N.}
\newblock Tight error bounds of geometric problems on convex objects with
  imprecise coordinates.
\newblock In {\em Jap.\ Conf.\ on Discrete and Comput.\ Geom.\/} (2000), LNCS
  2098, pp.~252--263.

\bibitem{obj-ue-05}
{\sc Ostrovsky-Berman, Y., and Joskowicz, L.}
\newblock Uncertainty envelopes.
\newblock In {\em 21st European Workshop on Comput.\ Geom.\/} (2005),
  pp.~175--178.

\bibitem{Phi08}
{\sc Phillips, J.~M.}
\newblock Algorithms for $\eps$-approximations of terrains.
\newblock In {\em ICALP\/} (2008).

\bibitem{Phi09}
{\sc Phillips, J.~M.}
\newblock {\em Small and Stable Descriptors of Distributions for Geometric
  Statistical Problems}.
\newblock PhD thesis, Duke University, 2009.

\bibitem{1644250}
{\sc Sarma, A.~D., Benjelloun, O., Halevy, A., Nabar, S., and Widom, J.}
\newblock Representing uncertain data: models, properties, and algorithms.
\newblock {\em The VLDB Journal 18}, 5 (2009), 989--1019.

\bibitem{spencer}
{\sc Spencer, J., Srinivasan, A., and Tetai, P.}
\newblock Discrepancy of permutation families.
\newblock Unpublished manuscript, 2001.

\bibitem{TCXNKP05}
{\sc Tao, Y., Cheng, R., Xiao, X., Ngai, W.~K., Kao, B., and Prabhakar, S.}
\newblock Indexing multi-dimensional uncertain data with arbitrary probability
  density functions.
\newblock In {\em VLDB\/} (2005).

\bibitem{MDFW00}
{\sc van~der Merwe, R., Doucet, A., de~Freitas, N., and Wan, E.}
\newblock The unscented particle filter.
\newblock In {\em NIPS\/} (2000), vol.~8, pp.~351--357.

\bibitem{kl-lbbsd-10}
{\sc van Kreveld, M., and L{\"o}ffler, M.}
\newblock Largest bounding box, smallest diameter, and related problems on
  imprecise points.
\newblock {\em Computational Geometry: Theory and Applications 43\/} (2010),
  419--433.

\bibitem{VC71}
{\sc Vapnik, V., and Chervonenkis, A.}
\newblock On the uniform convergence of relative frequencies of events to their
  probabilities.
\newblock {\em Theory of Probability and its Applications 16\/} (1971),
  264--280.

\bibitem{aggregate}
{\sc Yang, S., Zhang, W., Zhang, Y., and Lin, X.}
\newblock Probabilistic threshold range aggregate query processing over
  uncertain data.
\newblock {\em Advances in Data and Web Management\/} (2009), 51--62.

\bibitem{ZLTZW12}
{\sc Zhang, Y., Lin, X., Tao, Y., Zhang, W., and Wang, H.}
\newblock Efficient computation of range aggregates against uncertain location
  based queries.
\newblock {\em IEEE Transactions on Knowledge and Data Engineering 24\/}
  (2012), 1244--1258.

\end{thebibliography}
\bibliographystyle{acm}

\appendix
\section{Low Discrepancy to $\eps$-Coreset}
\label{app:MR}

Mainly in the 90s Chazelle and Matousek~\cite{CW89,CM96,divide,Mat99,Cha01} led the development of method to convert from a low-discrepancy coloring to a coreset that allowed for approximate range queries.  Here we summarize and generalize these results.

We start by restating a results of Phillips~\cite{Phi08,Phi09} which generalizes these results, here we state it a bit more specifically for our setting.  

\begin{theorem}[Phillips~\cite{Phi08,Phi09}]
Consider a point set $P$ of size $n$ and a family of subsets $\Eu{A}$.
Assume an $O(n^\beta)$ time algorithm to construct a coloring $\chi : P \to \{-1,+1\}$ so $\disc_\chi(P,\Eu{A}) = O(\gamma \log^\omega n)$ where $\beta$, $\gamma$, and $\omega$ are constant algorithm parameters dependent on $\Eu{A}$, but not $P$ (or $n$).  
There exists an algorithm to construct an $\eps$-sample of $(P,\Eu{A})$ of size $g(\eps,\Eu{A}) = O((\gamma/\eps) \log^\omega(\gamma/\eps))$ in time 
$O(n \cdot g(\eps,\Eu{A})^{\beta-1})$.  
\label{thm:disc2samp}
\end{theorem}

Note that we ignored non-exponential dependence on $\omega$ and $\beta$ since in our setting they are data and problem independent constants.  But we are more careful with $\gamma$ terms since they depend on $k$, the number of locations of each uncertain point.  

We restate the algorithm and analysis here for completeness, using $g = g(\eps,\Eu{A})$ for shorthand.  
Divide $P$ into $n/g$ parts $\{\bar P_1, \bar P_2, \ldots, \bar P_{n/g}\}$ of size $k = 4(\beta + 2) g$.  Assume this divides evenly and $n/g$ is a power of two, otherwise pad $P$ and adjust $g$ by a constant.  
Until there is a single set, repeat the following two stages.  
In stage 1, for $\beta+2$ steps, pair up all remaining sets, and for all pairs (e.g. $P_i$ and $P_j$) construct a low-discrepancy coloring $\chi$ on $P_i \cup P_j$ and discard all points colored $-1$ (or $+1$ at random).  In the $(\beta+3)$rd step pair up all sets, but do not construct a coloring and halve.  
That is every epoch ($\beta+3$ steps) the size of remaining sets double, otherwise they remain the same size.  
When a single set remains, stage 2 begins; it performs the color-halve part of the above procedure until $\disc(P,\Eu{A}) \leq \eps n$ as desired.  

We begin analyzing the error on a single coloring.  
\begin{lemma}
The set $P^+ = \{p \in P \mid \chi(p) = +1\}$ is an $(\disc_\chi(P,\Eu{A})/n)$-sample of $(P,\Eu{A})$.  
\label{lem:1round}
\end{lemma}
\begin{proof}
\begin{align*}
\max_{R \in \Eu{A}} \left| \frac{|P \cap R|}{|P|} - \frac{|P^+ \cap R|}{|P^+|} \right| 
= 
\max_{R \in \Eu{A}} \left| \frac{|P \cap R| - 2 |P^+ \cap R|}{n} \right| 
 \leq
\frac{\disc_\chi(P,\Eu{A})}{n}. 
\end{align*}
\end{proof}
We also note two simple facts~\cite{Cha01,Mat99}:
\begin{itemize}
\item[(S1)] If $Q_1$ is an $\eps$-sample of $P_1$ and $Q_2$ is an $\eps$-sample of $P_2$, then $Q_1 \cup Q_2$ is an $\eps$-sample of $P_1 \cup P_2$.  
\item[(S2)] If $Q$ is an $\eps_1$-sample of $P$ and $S$ is an $\eps_2$ sample of $Q$, then $S$ is an $(\eps_1+\eps_2)$-sample of $P$.  
\end{itemize}
Note that (S1) (along with Lemma \ref{lem:1round}) implies the arbitrarily decomposing $P$ into $n/g$ sets and constructing colorings of each achieves the same error bound as doing so on just one.  And (S2) implies that chaining together rounds adds the error in each round.  
It follows that if we ignore the $(\beta+3)$rd step in each epoch, then there is $1$ set remaining after $\log (n/g)$ steps.  The error caused by each step is $\disc(g,\Eu{A})/g$ so the total error is $\log (n/g) (\gamma \log^\omega g)/g = \eps$.  Solving for $g$ yields
$g = O(\frac{\gamma}{\eps} \log (\frac{n \eps}{\gamma}) \log^\omega(\frac{\gamma}{\eps}))$.  

Thus to achieve the result stated in the theorem the $(\beta+3)$rd step skip of a reduce needs to remove the $\log(n\eps/\gamma)$ term from the error.  This works!  After $\beta+3$ steps, the size of each set is $2g$ and the discrepancy error is $\gamma \log^\omega(2g) / 2g$.  This is just more than half of what it was before, so the total error is now:
\[\sum_{i=0}^{\frac{\log(n/g)}{\beta+3}} (\beta+3) \gamma \log^\omega(2^i g) / (2^i g) = \Theta(\beta (\gamma \log^\omega g)/g) = \eps. \]  Solving for $g$ yields $g = O(\frac{\beta \gamma}{\eps} \log^\omega(1/\eps))$ as desired.  
Stage 2 can be shown not to asymptotically increase the error.  

To achieve the runtime we again start with the form of the algorithm without the halve-skip on every $(\beta+3)$rd step.  Then the first step takes $O((n/g) \cdot g^{\beta})$ time.  And each $i$th step takes $O((n/2^{i-1}) g^{\beta-1})$ time.  Since each subsequent step takes half as much time, the runtime is dominated by the first $O(n g^{\beta-1})$ time step.  

For the full algorithm, the first epoch ($\beta+3$ steps, including a skipped halve) takes $O(n g^{\beta-1})$ time, and the $i$th epoch takes $O(n/2^{(\beta+2)i} (g 2^i)^{\beta-1}) = O(n g^{\beta-1} / 2^{3i})$ time.  Thus the time is still dominated by the first epoch.  Again, stage 2 can be shown not to affect this runtime, and the total runtime bound is achieved as desired, and completes the proof.  

Finally, we state a useful corollary about the expected error being $0$.  This holds specifically when we choose to discard the set $P^+$ or $P^- = \{p \in P \mid \chi(p) = -1\}$ at random on each halving.  

\begin{corollary}
\label{cor:E0}
The expected error for any range $R \in \Eu{A}$ on the $\eps$-sample $T$ created by Theorem \ref{thm:disc2samp} is 
\[
\E \left[ \frac{|R \cap P|}{|P|} - \frac{|T \cap R|}{|T|} \right]  = 0.
\]  
\end{corollary}
Note that there is no absolute value taken inside $\E[\cdot]$, so technically this measures the expected undercount.

\Paragraph{\RE-discrepancy}
We are also interested in achieving these same results for \RE-discrepancy.  To this end, the algorithms are identical.  Lemma \ref{lem:REcolor} replaces Lemma \ref{lem:1round}.  (S1) and (S2) still hold.  Nothing else about the analysis depends on properties of $\disc$ or \RE-$\disc$, so Theorem \ref{thm:disc2samp} can be restated for \RE-discrepancy.

\begin{theorem}
Consider an uncertain point set $P$ of size $n$ and a family of subsets $\Eu{A}$ of $P_\cert$.
Assume an $O(n^\beta)$ time algorithm to construct a coloring $\chi : P \to \{-1,+1\}$ so \RE-$\disc_\chi(P,\Eu{A}) = O(\gamma \log^\omega n)$ where $\beta$, $\gamma$, and $\omega$ are constant algorithm parameters dependent on $\Eu{A}$, but not $P$ (or $n$).  
There exists an algorithm to construct an $\eps$-\RE coreset of $(P,\Eu{A})$ of size $g(\eps,\Eu{A}) = O((\gamma/\eps) \log^\omega(\gamma/\eps))$ in time 
$O(n \cdot g(\eps,\Eu{A})^{\beta-1})$.  
\label{thm:RE-disc2samp}
\end{theorem}

\end{document}